\documentclass[comsoc,10pt]{IEEEtran}
\IEEEoverridecommandlockouts
\usepackage{cite}
\usepackage{amsmath,amssymb,amsfonts}
\usepackage{algorithmic}
\usepackage{graphicx}
\usepackage{textcomp}
\usepackage{xcolor}
\def\BibTeX{{\rm B\kern-.05em{\sc i\kern-.025em b}\kern-.08em
    T\kern-.1667em\lower.7ex\hbox{E}\kern-.125emX}}

\usepackage{datetime,comment}
\usepackage{url}
\usepackage{hyperref, color, amsmath}
\usepackage{graphicx}
\usepackage{wrapfig}
\usepackage{paralist,amsmath,amsthm}

\hypersetup{colorlinks=true,linkcolor=blue, citecolor=blue}

\usepackage{tikz}
\usepackage{pgfplots}
\usepackage{xspace}

\newcommand{\cPBFT}{PBFT\xspace}
\newcommand{\cLinearPBFT}{Linear-PBFT\xspace}
\newcommand{\cSBFTTagA}{Linear-PBFT+Fast Path\xspace}
\newcommand{\cSBFTCZero}{\sysname(c=0)\xspace}
\newcommand{\cSBFTCEight}{\sysname(c=8)\xspace}

\newcommand{\cSBFTCThree}{\sysname(c=3)\xspace}

\definecolor{blues1}{RGB}{198, 219, 239}
\definecolor{blues2}{RGB}{158, 202, 225}
\definecolor{blues3}{RGB}{107, 174, 214}
\definecolor{blues4}{RGB}{49, 130, 189}
\definecolor{blues5}{RGB}{8, 81, 156}
\definecolor{antiquefuchsia}{rgb}{0.57, 0.36, 0.51}
\definecolor{asparagus}{rgb}{0.53, 0.66, 0.42}
\definecolor{darkspringgreen}{rgb}{0.09, 0.45, 0.27}
\definecolor{darkslategray}{rgb}{0.18, 0.31, 0.31}
\definecolor{coralred}{rgb}{1.0, 0.25, 0.25}

\pgfplotsset{throughputGraphStyle/.style={%
		font=\tiny,
		width=0.29\paperwidth,
		xmin=4,xmax=256,
		ymin=0,
		enlargelimits=true,
		ymajorgrids=true,
		grid=major,
		grid style={dashed, gray!30},
		ylabel style={align=center},
		ylabel near ticks,
		xtick={4, 32,64, 128, 192, 256}, 
		scaled y ticks=true, 
		max space between ticks=10,
		enlargelimits=0.05,
		enlarge y limits==0.10,
}}

\pgfplotsset{LatencyGraphStyle/.style={%
		font=\tiny,
		width=0.29\paperwidth,
		enlargelimits=true,
		ymajorgrids=true,
		grid=major,
		grid style={dashed, gray!30},
		ylabel style={align=center},
		ylabel near ticks,
		scaled y ticks=false, 
}}

\pgfplotsset{plotSetA/.style={%
		legend columns=5,
		legend entries={\cPBFT, \cLinearPBFT, \cSBFTTagA, \cSBFTCZero, \cSBFTCEight },
		legend to name=plotSetALegened,
		legend style={font=\footnotesize, draw=none, text depth=0pt, legend columns=-1,},
}} 

\pgfplotsset{plotSetB/.style={%
		legend columns=5,
		legend entries={\cPBFT, \cLinearPBFT, \cSBFTTagA, \cSBFTCZero, \cSBFTCEight },
		legend to name=plotSetBLegened,
		legend style={font=\footnotesize, draw=none, text depth=0pt, legend columns=-1,},
}}

\pgfplotsset{plotSetC/.style={%
		legend columns=2,
		legend entries={\cSBFTCZero, \cSBFTCThree,  },
		legend to name=plotSetCLegened,
		legend style={font=\footnotesize, draw=none, text depth=0pt, legend columns=-1,},
}}

\newtheorem{theorem}{Theorem}[section]
\newtheorem{lemma}[theorem]{Lemma}

\newcommand\ls{\mathit{ls}}

\newcommand\window{\mathit{win}}

\newcommand\batch{\mathit{batch}}

\newcommand\CX{\mathit{LX}}

\newcommand\FX{\mathit{FX}}
\usepackage{xspace}
\newcommand\sysname{SBFT\xspace}

\title{\sysname: a Scalable and Decentralized Trust Infrastructure}

\date{}

\author{Guy Golan Gueta (VMware Research) \and
              Ittai Abraham (VMware Research) \and
              Shelly~Grossman~(TAU) \and
              Dahlia~Malkhi~(VMware Research) \and
              Benny Pinkas~(BIU) \and
              Michael~K.~Reiter~(UNC-Chapel Hill) \and
              Dragos-Adrian~Seredinschi~(EPFL) \and
              Orr~Tamir~(TAU) \and
              Alin~Tomescu~(MIT)}{ }

\begin{document}
\maketitle

\begin{abstract}
\sysname  is a state of the art Byzantine fault tolerant permissioned blockchain system that addresses the challenges of scalability, decentralization and world-scale geo-replication. \sysname is optimized for decentralization and can easily handle more than 200 active replicas in a real world-scale deployment. 

We evaluate \sysname in a world-scale geo-replicated deployment with 209 replicas withstanding f=64 Byzantine failures. We provide experiments that show how the different algorithmic ingredients of \sysname increase its performance and scalability. The results show that \sysname simultaneously provides almost 2x better throughput and about 1.5x better latency relative to a highly optimized system that implements the PBFT protocol.

To achieve this performance improvement, \sysname uses a combination of four ingredients: using collectors and threshold signatures to reduce communication to linear, using an optimistic fast path, reducing client communication and utilizing redundant servers for the fast path.
\end{abstract}


\section{Introduction}

Centralization often provides good performance, but has the drawback of posing a single point of failure \cite{S01} and economically, often creates monopoly rents and hampers innovation \cite{Dixon18}. The success of decentralization as in Bitcoin~\cite{N09} and Ethereum~\cite{W14} have spurred the imagination of many as to the potential benefits and significant potential value to society of a scalable decentralized trust infrastructure.

While fundamentally permissionless, the economic friction of buying and then running a Bitcoin or Ethereum mining rig has inherent economies of scale and unfair advantages to certain geographical and political regions. This means that miners are strongly incentivized to join a small set of large mining coalitions~\cite{bitcoinrace}.

In a 2018 study, Gencer et al. \cite{FC18} show that contrary to popular belief, Bitcoin and Ethereum are less decentralized than previously thought. Their study concludes that for both Bitcoin and Ethereum, the top $<20$ mining coalitions control over $90\%$ of the mining power. The authors comment ``These results show that a Byzantine quorum system of size 20 could
achieve better decentralization than Proof-of-Work mining at a much lower resource
cost''. This comment motivates the study of BFT replication systems that can scale to many replicas and are optimized for world scale wide area networks.

BFT replication is a key ingredient in consortium blockchains \cite{HL18,istanbul17}. In addition, large scale  BFT deployments are becoming an important component of public blockchains \cite{CV17}.
There is a growing interest in replacing or combining the current Proof-of-Work mechanisms  with Byzantine fault tolerant replication~\cite{V15,HL,EEA,VB, SBV17}.
Several recent proposals~~\cite{KJ16,M16,PS16,AMNRS16,algo17} explore the idea of building distributed ledgers
that elect a committee (potentially of a few tens or hundreds of nodes) from a large pool of nodes (potentially thousands or more) and have this smaller committee run a Byzantine fault tolerant replication protocol.
In all these protocols, it seems that to get a high security guarantee the size of the  committee needs to be such that it can tolerate at least tens of malicious nodes.

Scaling BFT replication to tolerate tens of malicious nodes requires to re-think BFT algorithms and re-engineer them for high scale. This is the starting point of our work.

\subsection{\sysname: a Scalable Decentralized Trust Infrastructure for Blockchains.}
The main contribution of this paper is a BFT system that
is optimized to work over a group of hundreds of replicas in a world-scale deployment. 
We evaluate our system, \sysname, in a world-scale geo-replicated deployment with 209 replicas withstanding f=64 Byzantine failures. We provide experiments that show how the different algorithmic ingredients of \sysname increase its performance and scalability. The results show that \sysname simultaneously provides almost 2x better throughput and about 1.5x better latency relative to a highly optimized system that implements the PBFT~\cite{CL99} protocol.

Indeed \sysname design starts with the PBFT~\cite{CL99} protocol and then proceeds to add four key design ingredients. Briefly, these ingredients are: (1) going from PBFT to linear PBFT; (2) adding a fast path; (3) using cryptography to allow a single message acknowledgement; (4) adding redundant servers to improve resilience and performance.
We show how each of the four ingredients improves the performance of \sysname. As we will discuss in detail, each ingredient is related to some previous work. The main contribution of \sysname is in the novel combination of these ingredients into a robust system.

\textbf{Ingredient 1: from PBFT to linear PBFT.}
Many previous systems, including PBFT~\cite{CL99}, use an all-to-all message pattern to commit a decision block. 
A simple way to reduce an all-to-all communication pattern to a linear communication pattern is to use a collector. Instead of sending to everyone, each replica sends to the collector and this collector broadcasts to everyone. We call this version \textit{linear PBFT}.
When messages are cryptographically signed, then using threshold signatures \cite{S00,CKS00} one can reduce the outgoing collector message size from linear to constant.

Zyzzyva \cite{KAD09} used this pattern to reduce all-to-all communication by pushing the collector duty to the clients. \sysname pushes the collector duty to the replicas in a round-robin manner. We believe that moving the coordination burden to the replicas is more suited to a blockchain setting where there are many light-weight clients with limited connectivity.
In addition, \sysname uses threshold signatures to reduce the collector  message size and the total computational overhead of verifying signatures. \sysname also uses a round-robin revolving collector to reduce the load. \sysname uses $c+1$ collectors (instead of one) to improve fault tolerance and handle $c$ slow or faulty collectors (where $c$ is typically a small constant).

\textbf{Ingredient 2: Adding a fast path.}
As in Zyzzyva \cite{KAD09}, \sysname allows for a faster agreement path in \textit{optimistic executions}: where all replicas are non-faulty and synchronous. 
No system we are aware of correctly incorporates a dual mode that allows to efficiently run either a fast path or a slow path. Previous systems tried to get a dual mode protocol to do the 
right thing, especially in the view-change protocol, but proved trickier than one
thinks~\cite{revisiting17,Zelma18}. We believe \sysname implements the first correct and practical dual mode view change. We 
rigorously analyzed and tested our view change protocol.  
We note that  Refined-Quorum-Systems \cite{RQS10}, a fast single shot Byzantine consensus protocol, does provide a correct dual model but its protocol for obtaining liveness seems to require exponential time computation in the worst case (and is just single-shot). Azyzzyva \cite{next15}, provides a fast path State-Machine-Replication protocol and allows to switch to a slow path, but avoids running both a fast path and a slow path concurrently. So switching between modes in Azyzzyva requires a lengthy view change each time for each switch, while \sysname can seamlessly switch between paths (without a view change).

\textbf{Ingredient 3: reducing client communication from \textit{f{+}1} to \textit{1}.}
Once threshold signatures are used then an obvious next step is to use them to reduce the number of messages a client needs to receive and verify.
In all previous solutions, including~\cite{CL99,KAD09,BFT-SMART14}, each client needs to receive at least $f{+}1 {=} O(n)$ messages,
each requiring a different signature verification
for request acknowledgement (where $f$ is the number of faulty replicas in a system with $n=3f+1$ replicas). When there are many replicas and many clients, this may add significant overhead.
In \sysname, in the common case, each client needs only one message, containing a single public-key signature for request acknowledgement.
This single message improvement means that \sysname can scale to support many extremely thin clients.

\sysname reduces the per-client linear cost to just one message by adding a phase that uses an \textit{execution collector} to aggregate the execution threshold signatures and send each client a single message carrying one signature. Just like public blockchains (Bitcoin and Ethereum), \sysname uses a Merkle tree to efficiently authenticate information that is read from just one replica.

\sysname uses Boneh–--Lynn--–Shacham (BLS) signatures~\cite{BLS04} which have security that is comparable to 2048-bit RSA signatures but are  just 33 bytes long. Threshold signatures~\cite{threshold-bls} are much faster when implemented over BLS (see Section~\ref{sec:crypto}).

\textbf{Ingredient 4: adding redundant servers to improve resilience and performance.}
\sysname is safe even if there are $f$ Byzantine failures, but the standard fast path works only if all replicas are non-faulty and the system is synchronous. So even a single slow replica may tilt the system from the fast path to the slower path.
To make the fast path more prevalent, 
\sysname allows the fast path to tolerate up to a small number $c$ (parameter) of
crashed or straggler nodes out of $n=3f+2c+1$ replicas. This approach follows the theoretical results that has been suggested before in the
context of single-shot consensus algorithms~\cite{MA06}. So \sysname only falls back to the slower path if there are more than $c$ (and for safety fewer than $f$) faulty replicas. In our experiments we found that setting  $c\leq f/8$ is a good heuristic for up to a few hundreds of replicas.

\subsection{Evaluating \sysname's scalability.}

We implemented \sysname as a scalable BFT engine and a blockchain that executes EVM smart contracts~\cite{W14} (see Section \ref{sec:implementation}).  

All our experimental evaluation is done in a setting that withstands $f=64$ Byzantine failures in a real Wide Area Network deployment.

We first conduct standard key-value benchmark experiments with synthetic workloads. We start with a scale optimized PBFT and then show how adding each ingredient helps improve performance.

While standard key-value benchmark experiments with synthetic workloads are a good way to compare the BFT engine internals, we realize that real world blockchains like Ethereum have a very different type of execution workload based on smart contracts.

We conduct experiments on real world workloads for our blockchain in order to measure the perfomance of a more realisic workload. Our goal is not to do a comparison of a permissioned BFT system against a permissionless proof-of-work system, this is clearly not a fair comparison.

We take 500,000 smart contract executions
that were processed by Ethereum during a 2 months period.
Our experiments show that 
in a world-scale geo-replicated deployment with 209 replicas withstanding f=64 Byzantine failures, we obtain throughput of over 170 smart contract transactions per second with average latency of 620 milliseconds. Our Experiments show that \sysname simultaneously provides almost 2x better throughput and about 1.5x better latency relative to a highly optimized system that implements the PBFT protocol.

We conclude that \sysname is more scalable and decentralized relative to previous BFT solutions. Relative to state-of-art proof-of-work systems like Ethereum, \sysname can run the same smart contracts at a higher throughput, better finality and latency, and can withstand $f=64$ colluding members out of more than 200 replicas. Clearly, Ethereum and other proof-of-work systems benefit from being an open permissionless system, while \sysname is a permissioned blockchain system. We leave the integration of \sysname in a permissionless system for future work. 

We are not aware of any other permissioned blockchain system that can be freely used and can be deployed in a world scale WAN that can scale to over 200 replicas and can withstand $f=64$ Byzantine failures.  We therefore spent several months significantly improving, fixing and hardening an existing PBFT code-base in order to make it reliably work in our experimental setting. We call this implementation \textit{scale optimized PBFT} and experimentally compare it to \sysname.

\textbf{Contributions.}
The main contribution of this paper is a BFT system that
is optimized to work over a group of hundreds of replicas and supports the execution of modern EVM smart contract executions in a world-scale geo-distributed deployment. \sysname obtains its scalability via a combination of 4 algorithmic ingredients: (1) using a collector to obtain linear communication, (2) adding a fast path with a correct view change protocol, (3) reducing client communication using collectors and threshold signatures, (4) adding redundant servers for resilience and performance. While our view change protocol is new, one could argue that the other ingredients mentioned have appeared in some form in previous work. Nevertheless, \sysname is the first to careful weave and implement all these ingredients into a highly efficient and scalable BFT system.

\section{System Model}

We assume a standard \textit{asynchronous} BFT system model where an adversary can control up to $f$ Byzantine nodes and can delay any message in the network by any finite amount (in particular we assume a re-transmit layer and allow the adversary to drop any given packet a finite number of times).
To obtain liveness and our improved results we also distinguish two special conditions. We say that the system is in the \textit{synchronous mode},  when the adversary can control up to $f$ Byzantine nodes, but messages between any two non-faulty nodes have a known bounded delay. Finally we say that the system is in the  \textit{common mode},  when the adversary only controls $c\leq f$ nodes that can only crash or act slow, and messages between any two non-faulty nodes have a known bounded delay.
This three-mode model follows that of Parameterized FaB Paxos~\cite{MA06}.

For $n=3f+2c+1$ replicas \sysname obtains the following  properties: 

(1) \textit{Safety} in the standard asynchronous model (adversary controlling at most $f$ Byzantine nodes and all network delays). This means that any two replicas that execute a decision block for a given sequence number,
execute the same decision block.

(2) \textit{Liveness} in the synchronous mode (adversary controlling at most $f$ Byzantine nodes). Roughly speaking, liveness means that client requests return a response.

(3) \textit{Linearity} in the common mode  (adversary controlling at most a constant $c$ slow/crashed nodes).
Linearity means that in an abstract model where we assume the number of operations in a block is $O(n)$ and we assume the number of clients is also $O(n)$,
then the amortized cost to commit an operation is a constant number of constant size messages. 
In more practical terms, Linearity means that committing each block takes a linear number of constant size messages and that each client sends and receives receives just one message per operation.

\section{Modern Cryptography}
\label{sec:crypto}
We use threshold signatures, where for a threshold parameter $k$, any subset of $k$ from a total of $n$ signers can collaborate to produce a valid signature on any given message, but no subset of less than $k$ can do so.
Threshold signatures have proved useful in previous BFT algorithms and systems (e.g.,~\cite{S00,CKS00,ADKLDNOZ06, ADDKLNOZ10}).  
Each signer holds a distinct private signing key that it can use to generate a signature share. We denote by $x_i(d)$ the signature share on digest $d$ by signer $i$. Any $k$ valid signature shares $\{x_{j}(d) \mid j \in J, |J|=k\}$ on the same digest $d$ can be combined into a single signature 
$x(d)$ using a public function, yielding a digital signature $x(d)$.  A verifier can verify this
signature using a single public key. We use threshold signature schemes which are \textit{robust}, meaning signers can
efficiently filter out invalid signature shares from malicious participants.

We use a robust threshold signature scheme based on Boneh–--Lynn--–Shacham (BLS) signatures~\cite{BLS04}. 
While RSA signatures are built in a group of hidden order, BLS are built using \textit{pairings}~\cite{triple-dh} over elliptic curve groups of known order.
Compared to RSA signatures with the same security level, BLS signatures are substantially shorter. BLS requires 33 bytes compared to 256 bytes for 2048-bit RSA. Creating and combining RSA signature shares via interpolation in the exponent requires several expensive operations~\cite{S00}. In contrast,  BLS threshold signatures~\cite{threshold-bls} allow straightforward interpolation in the exponent with no additional overhead. Unlike RSA, BLS signature shares support batch verification, allowing multiple signature shares (even of different messages) to be validated at nearly the same cost of validating only one signature~\cite{threshold-bls}.

We assume a computationally bounded adversary that cannot do better than known attacks as of 2018 on the cryptographic hash function SHA256 and on BLS BN-P254~\cite{bn-p254} based signatures.
We use a PKI setup between clients and replicas for authentication.

\section{Service Properties}
\sysname provides a scalable fault tolerant implementation of a generic replicated service (i.e., a \textit{state machine replication} service). On top of this we implement an \textit{authenticated key-value }store that uses a Merkle tree interface~\cite{M87} for data authentication. On top of this we implement a 
smart contract layer capable of running EVM byte-code. This layered architecture allows us in the future to integrate other smart contract languages by simply connecting them to the generic authenticated key-value store and allow for better software reuse.

\textbf{Generic service.}
As a generic replication library, \sysname requires an implementation of the following service interface to be received as an initialization parameter. 
The interface implements any deterministic replicated \textit{service} with \textit{state}, deterministic \textit{operations} and read-only \textit{queries}.  
An execution $val=execute(\mathcal{D},o)$ modifies state $\mathcal{D}$ according to the operation $o$ and returns an output $val$. A query $val=query(\mathcal{D},q)$ returns the value of the query $q$ given state $\mathcal{D}$ (but does not change state $\mathcal{D}$). These operations and queries can perform arbitrary deterministic computations on the state. 

The state of the service moves in discrete blocks. Each block contains a series of requests. We denote by $\mathcal{D}_j$ the state of the service at the end of sequence number $j$. We denote by $req_j$ the series of operations of block $j$, that changes the state from state $\mathcal{D}_{j-1}$ to state $\mathcal{D}_j$.

\textbf{An authenticated key-value store.}
For our blockchain implementation we use a key-value store. 
In order to support efficient client acknowledgement from one replica we augment our key-value store with a \textit{data authentication} interface. As in public permissionless blockchains, 
we use a Merkle trees interface \cite{M87} to authenticate data. 
To provide data authentication we require an implementation of the following interface:
(1) $d=\mathit{digest}(\mathcal{D})$ returns the Merkle hash root of $\mathcal{D}$ as digest.
(2) $P=\mathit{proof}(o,l, s,\mathcal{D}, val)$ returns a proof that operation $o$ was executed as the $l$th operation in the  series of requests in the decision block whose sequence number is $s$, whose state is $\mathcal{D}$ and the output of this operation was $val$.  
For a key-value store, proof for a $\mathit{put}$ operation is a Merkle tree proof that the $\mathit{put}$ operation was conducted as the $l$th operation in the requests of sequence number $s$.
For a read only-query $q$, we write $P=\mathit{proof}(q,s,\mathcal{D},val)$ and assume all such queries are executed with respect to $\mathcal{D}_{s}$ (the state $\mathcal{D}$ after completing sequence number $s$). 
For a key-value store, proof for a $\mathit{get}$ operation is a Merkle tree proof that at the state with sequence number $s$ the required variable has the desired value. (3) $\mathit{verify}(d,o,val,s, l,P)$ returns true iff $P$ is a valid proof that $o$ was executed as the $l$th operation in sequence number $s$ and the resulting state after this decision block was executed has a digest of $d$ and $val$ is the return value for operation $o$ (and similarly $\mathit{verify}(d,q,val,s,P)$ when $q$ is a query).
For a key-value store and a  $\mathit{put}$ operation above, the verification is the Merkle proof verification \cite{M87}  rooted at the digest $d$ (Merkle hash root).

\textbf{A Smart contract engine.}
We build upon the replicated key-value store a layer capable of executing Ethereum smart contracts.
This layered architecture allows us in the future to integrate other smart contract languages by simply connecting them to the generic authenticated key-value store and allow for better software reuse.
The EVM layer consists of two main components: 
(1) An implementation of the Ethereum Virtual Machine (EVM), which is the runtime engine of contracts;
(2) An interface for modeling the two main Ethereum transaction types (contract creation and contract execution) as operations in our replicated service.
Ethereum contracts are written in a language called \textit{EVM bytecode}~\cite{W14}, 
a Turing-complete stack-based low-level language, with special commands designed for the Ethereum platform.
The key-value store keeps the state of the ledger service. In particular, it saves the code of the contracts and the contracts' state.
The fact that EVM bytecode is deterministic ensures that the new state digest will be equal in all non-faulty replicas.

\section{\sysname Replication Protocol}
We maintain $n=3f+2c+1$ replicas where  each replica has a unique identifier in $\{1,\dots,3f+2c+1\}$.  This identifier is used to determine the threshold signature in the three threshold signatures: $\sigma$ with threshold ($3f+c+1$), $\tau$ with threshold  ($2f+c+1$), and $\pi$ with threshold ($f+1$). 

We adopt the approach of ~\cite{OL88,CL99} where replicas move from one \textit{view} to another using a \textit{view change} protocol. In a view, one replica is a \textit{primary} and others are backups. The primary is responsible for initiating decisions on a sequence of decisions.
Unlike PBFT~\cite{CL99}, some backup replicas can have additional roles as Commit collectors and/or Execution collectors. 
In a given view and sequence number, $c+1$ non-primary replicas are designated to be \textit{C-collectors} (Commit collectors) and $c+1$ non-primary replicas are designated to be \textit{E-collectors} (Execution collectors). These replicas are responsible for collecting threshold signatures, combining them and disseminating the resulting signature. For liveness, a single correct collector is needed. We use $c+1$ collectors for redundancy in the fast path (inspired by RBFT~\cite{RBFT13}). This increases the worst case message complexity to $O(cn)=O(n)$ when we assume $c$ is a small constant (for $n \approx 200$ we set $c=0,1,2,8$ with $f=64$). In practice we stagger the collectors, so in most executions just one collector is active and the others just monitor in idle.

Roughly speaking the algorithm works as follows in the \textit{fast path} (see Figure \ref{fig:SBFT} for $n=4$, $f=1$, $c=0$):

(1) Clients send operation \textit{request} to the primary.

(2) The primary gathers client requests, creates a decision block and forwards this block to the replicas as a \textit{pre-prepare} message.

(3) Replicas sign the decision block using their 
$\sigma$ ($3f+c+1$)-threshold signature and send a \textit{sign-share} message to the C-collectors.

(4) Each C-collector gathers the signature shares, creates a succinct \textit{full-commit-proof} for the decision block and sends it back to the replicas.
This single message commit proof
has a fixed-size overhead, contains a single signature 
and is sufficient for replicas to commit.

Steps (2), (3) and (4) require linear message complexity (when $c$ is constant) and replace the quadratic message exchange of previous solutions. By choosing a different C-collector group for each decision block, we balance the load over all replicas.

Once a replica receives a commit proof it commits the decision block. The replica then starts the \textit{execution protocol}: 

(1) When a replica has finished executing the sequence of blocks preceding the committed decision block, it executes the requests in the decision block  and signs a digest of the new state using its 
$\pi$ $(f+1)$
threshold signature, and sends a \textit{sign-state} message to the E-collectors.

(2) Each E-collector gathers the signature shares, and creates a succinct \textit{full-execute-proof} for the decision block. It then sends a certificate back to the replicas indicating the state is durable and a certificate back to the client indicating that its operation was executed.

This single message
has fixed-size overhead, contains a single signature and
is sufficient for acknowledging individual clients requests.

Steps (1) and (2) provide  single-message per-request acknowledgement for each client. All previous solutions required a linear number of  messages per-request acknowledgement for each client. When the number of clients is large this is a significant advantage.

Once again, by choosing a different E-collector group for each 
decision block, we spread the overall load of primary leadership, C-collection, and E-collection, among all the replicas.

\begin{figure}[h!]
        \centering
        \includegraphics[width=0.50\textwidth]{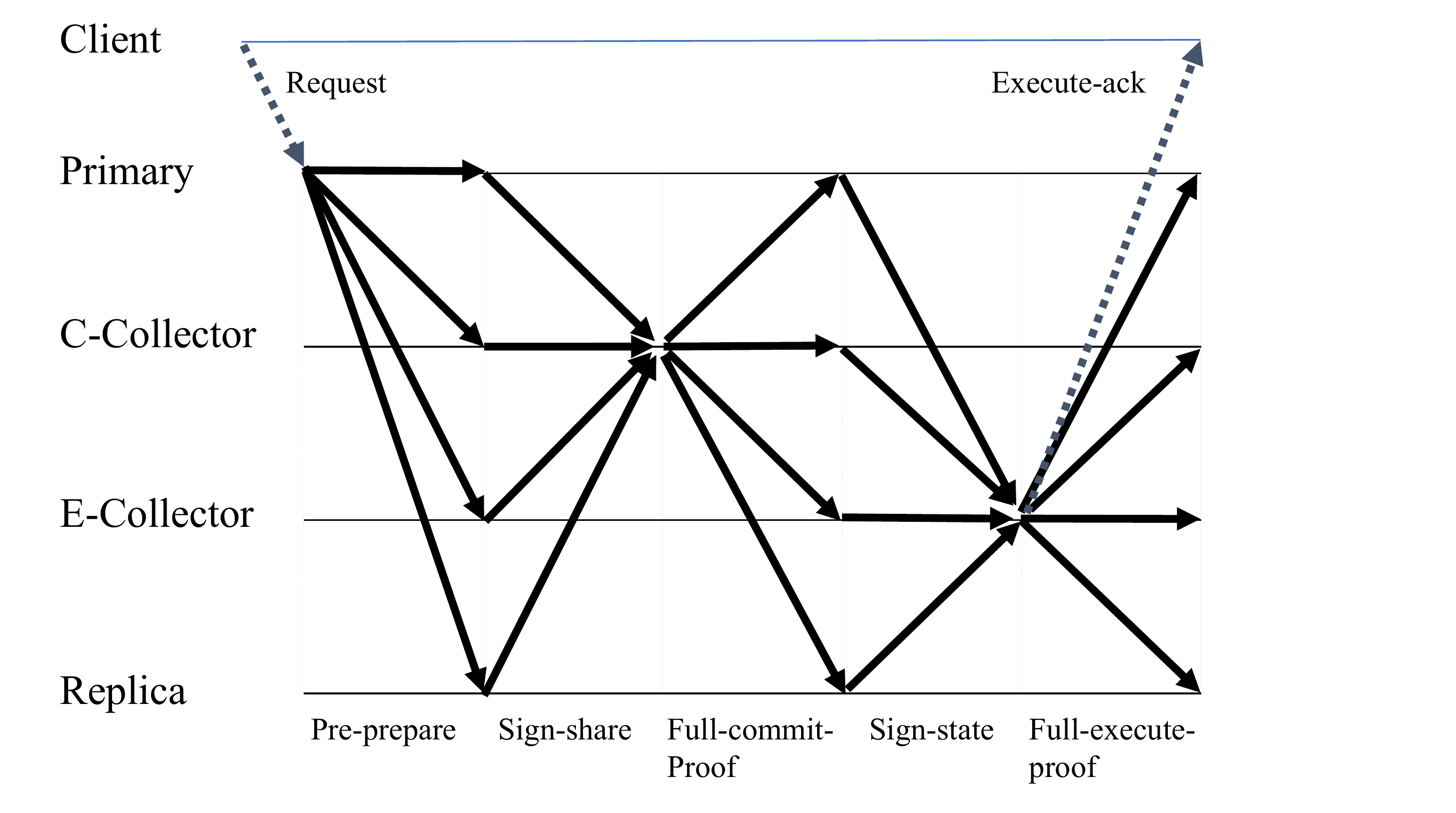}
        \caption{Schematic message flow for $n{=}4,f{=}1$, $c{=}0$. }     
        \label{fig:SBFT}
\end{figure}

\subsection{The Client}
Each client $k$ maintains a strictly monotone timestamp $t$ and requests an operation $o$ by sending a message $\langle \text{``request''}$, $o$, $t$, $k \rangle$ to what it believes is the primary.
The primary then sends the message to all replicas and replicas then engage in an agreement algorithm. 

Previous systems required clients to wait for $f+1$ replies to accept an execution acknowledgment. In our algorithm  the client waits for just a \textit{single} reply $\langle \text{``execute-ack''}$, $s$, $val$, $o$, $\pi(d)$, $\mathit{proof}(o,l,s,\mathcal{D},val) \rangle$
from one of the replicas, and accepts $val$ as the response from executing $o$ by verifying that $\mathit{proof}(o,l,s,\mathcal{D},val)$ is a proof that $o$ was executed 
as the $l$th operation of the decision block that resulted in the state whose sequence number is $s$, the return value of $o$ was $val$, the digest of $\mathcal{D}_{s}$  is $d$. This is done by checking the Merkle proof  $\mathit{verify}(d, o, val, s,l, \mathit{proof}(o,l,s,\mathcal{D},val)) = true$ and that $\pi(d)$ is a valid signature for $\mathcal{D}_{s}$
(when $o$ is long we just send the digest of $o$). 

Upon accepting an execute-ack message the client marks $o$ as executed and sets $val$ as its return value. 

As in previous protocols, if a client timer expires before receiving an execute-ack, the client resends the request to all replicas (and requests a PBFT style $f+1$ acknowledgement path).

\subsection{The Replicas}
The state of each replica includes a log of accepted messages sorted by sequence number, view number and message type. 
The state also includes the current view number, the last stable sequence number $\ls$ (see Section~\ref{sec:garbage}),  the state of the service $\mathcal{D}$ after applying all the committed requests. We also use a known constant $\window$ that limits the number of outstanding blocks.

Each replica has an identity $i \in \{1,\dots,n\}$ used to determine the value of the three threshold key shares:  $\sigma_i$ for a $3f+c+1$ threshold scheme, $\tau_i$ for a $2f+c+1$ threshold scheme, and $\pi_i$ for a $f+1$ threshold scheme. All messages between replicas are done using authenticated point-to-point channels (in practice using TLS 1.2).

As detailed below, replicas can have additional roles of being a \textit{primary} (Leader), a \textit{C-collector} (Commit collector) or an \textit{E-collector} (Execution collector).

The primary for a given view is chosen in a round robin way as a function of $view$. It also stores a current sequence number

The C-collectors and E-collector for a given view and sequence number are chosen as a pseudo-random group from all non-primary replicas, as a function of the sequence number and view\footnote{Randomly choosing the primary and the collectors to provide resilience against a more adaptive adversary is doable, but not is part of the current implementation}. For the fall back Linear-PBFT protocol we always choose the primary as the last collector.

The role of a C-collector is to collect commit messages and send a combined signature back to replicas to confirm commit. The role of an E-collector is to collect execution messages and send a combined signature back to replicas and clients so they all have a certificate that their request is executed.


\subsection{Fast Path}
The fast path protocol is the default mode of execution. It is guaranteed to make progress  when the system is synchronous and there are at most $c$ crashed/slow replicas. 

To commit a new decision block the primary starts a three phase protocol: \textit{pre-prepare, sign-share, commit-proof}. 
In the \textit{pre-prepare} phase the primary forwards its decision block to all replicas. In the \textit{sign-share} phase, each replica signs the requests using its threshold signature and sends it to the C-collectors. In the \textit{commit-proof} phase, each C-collector generates a succinct signature of the decision and sends it to all replicas.

\textbf{Pre-prepare phase:} The primary accepts
$\langle \text{``request''}, o,t,k \rangle$ from client $k$ if the operation $o$ passes the static service authentication and access control rules. Note that this is a state independent test which can be changed via a reconfiguration view change.

Upon accepting at least $b\geq \batch$ client messages (or reaching a timeout) it sets $r = (r_1,\dots,r_{b})$ and broadcasts $\langle \text{``pre-prepare''}, s, v, r \rangle$ to all $3f+2c+1$ replicas where $s$ is the current sequence number, and $v$ is the view number.

The parameter $\batch$ is set via an adaptive algorithm. Roughly speaking the value of $\batch$ is set to be the average number of pending requests divided by the half the maximum amount of allowed concurrent blocks (this number number was to 4 in the experiments).

\textbf{Sign-share phase:} A replica accepts $\langle \text{``pre-prepare''}, s, v, r \rangle$ from the primary if (1) 
its view equals $v$; (2) no previous \text{``pre-prepare''} with the sequence $s$ was accepted for view $v$; (3) the sequence number $s$ is between   $\ls$ and $\ls + \window$; (4)  $r$ is a valid series of operations that pass the authentication and access control requirements.

Upon accepting a pre-prepare message, replica $i$ computes $h=H(s||v||r)$ where $H$ is a cryptographic hash function (SHA256) then signs $h$ by computing a verifiable threshold signature $\sigma_{i}(h)$ 
and sends $\langle \text{``sign-share''}, s, v, \sigma_{i}(h) \rangle$ to the set of C-collectors $C\text{-}collectors(s,v)$.

\textbf{Commit-proof phase:}
a C-collector for $(s,v)$ accepts a $\langle \text{``sign-share''}, s, v, \sigma_{i}(h) \rangle$  from a replica $i$ if  (1) 
its view equals $view$; (2) no previous  \text{``sign-share''} with the same sequence $s$ has been accepted for this view
from replica $i$;
(3) the verifiable threshold signature $\sigma_{i}(h)$ passes the verification.

Upon a C-collector accepting $3f+c+1$ distinct sign-share messages it forms a combined signature $\sigma(h)$, and then sends $\langle \text{``full-commit-proof''}, s, v ,\sigma(h) \rangle$  to all replicas. 

\textbf{Commit trigger:} 
a replica accepts $\langle \text{``full-commit-proof''}, s, v, \sigma(h) \rangle$ if it accepted
 $\langle \text{``pre-prepare''}, s, v, r, h \rangle$, $h=H(s||v||r)$ and $\sigma(h)$ is a valid signature for $h$. Upon accepting a full-commit-proof message, the replica commits $r$ as the requests for sequence $s$.


\subsection{Execution and Acknowledgement} \label{sec:exec}

The main difference of our execution algorithm from previous work is the use of threshold signatures and single client responses.
Once a replica has a consecutive sequence of committed decision blocks it participates in a two phase protocol: \textit{sign-state, execute-proof}.

Roughly speaking, in the sign-state phase each replica signs its state using its $f+1$ threshold signature and sends it to the E-collectors. In the execute-proof phase, each E-collector generates a succinct execution certificate. It then sends this certificate back to the replicas and also sends each client a certificate indicating its operation(s) were executed.

\textbf{Execute trigger and sign state:} when all decisions up to sequence $s$ are executed, and $r$ is the committed request block for sequence $s$, then replica $i$ updates its state to $\mathcal{D}_{s}$ by executing the requests $r$ 
sequentially on the state $\mathcal{D}_{s-1}$. 

Replica $i$ then updates its digest on the state to $d=\mathit{digest}(\mathcal{D}_{s})$, signs $d$ by computing $\pi_{i}(d)$ 
and sends $\langle \text{``sign-state''}, s, \pi_{i}(d) \rangle$ to the set of E-collectors $E\text{-}collectors(s)$.

\textbf{Execute-proof phase:}
an E-collector for $s$ accepts a  $\langle \text{``sign-state''}, s, \pi_{i}(d) \rangle$  from a replica $i$ if
$\pi_i(d)$ passes the verification test.

Upon accepting $f+1$ sign-state messages, it combines them into a single signature $\pi(d)$ and sends $\langle \text{``full-execute-proof''}, s ,\pi(d) \rangle$  to all replicas. 
Replicas that receive full-execute-proof messages verify the signature to accept.

Then the E-collector, for each request $o \in r$ at position $l$ sends to the client $k$ that issued $o$ an execution acknowledgement, $\langle \text{``execute-ack''}, s, l, val,o ,\pi(d), \mathit{proof}(o,l,s,\mathcal{D},val) \rangle$, where $val$ is the response to $o$, $\mathit{proof}(o,l,s,\mathcal{D},val)$ is a proof that $o$ was executed and $val$ is the response at the state whose digest is from $\mathcal{D}_{s}$ and $\pi(d)$ is a signature that the digest of $\mathcal{D}_{s}$ is $d$.

The client, accepts 
$\langle \text{``execute-ack''}, s, l, val ,o,\pi(d),P \rangle$ if $\pi(d)$ is a valid signature and $\mathit{verify}(d, o, val, s,l,P) = true$.

Upon accepting an execute-ack message the client marks $o$ as executed and sets $val$ as its return value.  If the client timer expires then the client re-tries requests and asks for a regular PBFT style acknowledgement from $f+1$. 



\subsection{Linear-PBFT}
This is a fall-back protocol that can provide progress when the fast path cannot make progress.
This protocol is an adaptation of PBFT that is optimized to use threshold signatures and linear communication, avoiding all-to-all communication by using the primary as collector in the intermediate stage and as a fallback collector for commitment collection and execution collection. To guarantee progress when the primary is non-faulty, we use $c+1$ collectors and stagger the collectors so that the $c+1$st collector to activate is always the primary. The worst case communication is $O(cn)$ which is $O(n)$ when $c$ is a constant (say $c=2$). In particular choosing $c=0$ for the fall back protocol would  guarantee $O(n)$ messages (one can still have more collectors in the fast path).

In linear-PBFT, instead of broadcasting messages to all replicas, we use the primary as a single collector (or use $c+1$ collectors for a small constant $c \leq 2$) that composes the threshold signatures into a single signature message. This reduces the number of messages and public key operations to linear, and  makes each message contain just one public-key signature. We call this operation \emph{broadcast-via-collector}: each replica sends its message only to the $c+1$ collectors, each collector waits to aggregate a threshold signature and then sends it to all replicas.

\textbf{Sign-share phase:} we modify the sign-share message of replica $i$ to include both $\sigma_i(h)$ (needed for the fast path) and $\tau_i(h)$ (needed for the Linear PBFT path).

\textbf{Trigger for Linear-PBFT:} A C-collector (including the primary) that received enough threshold shares (via sign-share messages) to create $\tau(h)$ but not to create $\sigma(h)$ waits for a timeout to expire before sending a prepare message to all: $\langle \text{``prepare''}, s, v, \tau(h) \rangle$. This timer controls how long to wait for the fast path before reverting to the PBFT path, we use an adaptive protocol based on past network profiling to control this timer.

\textbf{Prepare phase:} Replica $i$ accepts $\langle \text{``prepare''}, s, v, \tau(h) \rangle$ if
(1)  its view equals $v$; (2)  no previous  \text{``prepare''} with sequence $s$ has been accepted for this view by $i$;
(3)  $\tau(h)$ passes its verification.
Replica $i$ sends $\langle \text{``commit''}, s, v, \tau_i(\tau(h)) \rangle$ to all the collectors. 

\textbf{PBFT commit-proof phase:} A C-collector (including the primary) that received enough threshold shares to create $\tau(\tau(h))$ sends a full-commit-proof-slow message to all: $\langle \text{``full-commit-proof-slow''}, s,v, \tau(\tau(h))\rangle$.

\textbf{Commit trigger for Linear-PBFT:} If a replica receives 
$\langle \text{``full-commit-proof-slow''}, s, v, \tau(\tau(h)) \rangle$ and $\langle \text{``pre-prepare''}, s, v, r, h \rangle$  it verifies that  $h=H(s||v||r)$ then commits $r$ as the decision block at sequence $s$.

\subsection{Garbage Collection and Checkpoint Protocol}\label{sec:garbage}

A decision block at sequence $s$ can have three states: (1) Committed - when at least one non-faulty replica has committed $s$; (2) Executed - when at least one non-faulty replica has committed all blocks from 1 to $s$; (3) Stable - when at least $f+1$ non-faulty replicas have executed $s$. 

When a decision block at sequence $s$ is stable we can garbage collect all previous decisions. As in PBFT we periodically (every $\window/2$) execute a checkpoint protocol in order to update $\ls$ the \textit{last stable} sequence number. 

To avoid the overhead of the quadratic PBFT checkpoint protocol, the second way to update $\ls$ is to add the following restriction. A replica only participates in a fast path of sequence $s$ if $s$ is between $\mathit{le}$ and $ \mathit{le} + (\window$/4) where $\mathit{le}$ is the last executed sequence number. With this restriction, when a replica commits in the fast path on $s$ it sets $\ls:=\max  \{ ls,s- (\window/4)\}$.


\subsection{View Change Protocol}\label{sec:view-change}

The view change protocol  handles 
the non-trivial complexity of having two commit modes: Fast-Path and Linear-PBFT.
Protocols having  two modes like \cite{MA06,KAD09, RQS10,next15} have to carefully handle cases where both modes provide a value to adopt and must explicitly choose the right one. 
\sysname implements a new view change protocol that maintains both safety and liveness while handling the challenges of two concurrent modes. 
\sysname's view change has been carefully implemented \footnote{\sysname has been actively developed and hardened for over 2 years.}, rigorously analyzed and tested\footnote{We ran experiments with hundreds of replicas, doing tens of thousands of view changes, and have tests for Primaries sending partial, equivocating and/or stale information.}.

\textbf{View change trigger:} a replica triggers a view change when a timer expires or if it receives a proof that the primary is faulty (either via a publicly verifiable contradiction or when $f+1$ replicas  complain).

\textbf{View-change phase:} Each replica $i$ maintains a variable 
$\ls$ which is the last stable sequence number. It prepares values $x_{\ls},x_{\ls+1},\dots,x_{\ls+\window}$ as follows. Set $x_{\ls}=\pi(d_{\ls})$ to be the signed digest on the state whose sequence is $\ls$.
For each $ \ls < j \le \ls+\window$ set $x_j=(lm_j,fm_j)$ to be a pair of values as follows: 

Set $lm_j$ to be $\tau(\tau(h_j))$ if a full-commit-proof-slow was accepted for sequence $j$; otherwise set $lm_j$ to be $(\tau(h_j),v_j)$ where $v_j$ is the highest view for sequence $j$ for which $2f+c+1$ prepares were accepted with hash $h_j$ in view $v_j$; otherwise  set $lm_j :=\text{``no commit``}$.

Set $fm_j$ to be $\sigma(h_j)$ if a full-commit-proof was accepted for sequence $j$; otherwise set $fm_j$ to be $(\sigma_i(h_j),v_j)$ where $v_j$ is the highest view for sequence $j$ for which a pre-prepare was accepted with hash $h_j$ at view $v_j$; otherwise set $fm_j :=\text{``no pre-prepare''}$.

Replica $i$ sends to the new primary of view $v+1$ the message $\langle \text{``view-change''}, v, \ls, x_{\ls}, x_{\ls+1},\dots,x_{\ls+\window} \rangle$
where $v$ is the current view number and $x_{\ls},\dots,x_{\ls+\window}$ as defined above.

\textbf{New-view phase:} 
The new primary gathers $2f+2c+1$ view change messages from replicas.  The new primary initiates a new view by sending a set of $2f+2c+1$ view change messages.

\textbf{Accepting a New-view:}
When a replica receives a set $I$ of $|I| = 2f+2c+1$ view change message it processes slots one by one. It starts with $\ls$, the highest valid stable sequence number in all view-change messages, and goes up to $\ls + \window$. For each such slot, a replica either decides it can commit a value, or it adopts it as a pre-prepare by the new primary, according to the algorithm below.

If a replica receives $\sigma(\star)$ or $\tau(\tau(\star)))$, it decides it. Else, it adopts a safe value: 

\textbf{Safe values:} A value $y$ is safe for sequence slot if the only safe thing for the new primary to do is to propose $y$ for the sequence slot in the new view. 

Roughly speaking, $y$ will be the value that is induced by the \textit{highest} view for which there is a potential value that could have been committed in a previous view. Defining this requires carefully defining the highest view for which there is a value in each of the two commit paths and then taking the highest view between the two paths. If there is no value that needs to be adopted, we fill the sequence with a special no-op operation.

More precisely, computing $y$ given $I$ is done as follows:

Set $\ls$ to be the highest last stable value $\ls_i$ sent in $I$ such that $i$ sent $\pi(d_{\ls_i})$ which is correct (this is a proof that $\ls_i$ is a valid checkpoint). 
Fix a slot $j$ within the range $[\ls .. (\ls+\window)]$. Let $X=\{x^i\}_{i \in I}$ be the set of values by the members $I$ for the slot. Since each $x  \in X$ is a pair we split into two sets $X=(\CX,\FX)$. 
If a member in $I$ sent values only up to a lower sequence position, then we can simulate as if these missing values are $x=(\text{``no commit``},\text{``no pre prepare``})$.

If $\FX$ contains $\sigma(h)$ or $\CX$ contains $\tau(\tau(h))$ then let $y$ be $h$ and commit once the message is known; otherwise

(1) If $\CX$ contains at least one $\tau(h)$ then let $\tau(h^*)$ be the $\tau$ signature with the highest view $v^*$ in $\CX$ and let $req^*$ be the corresponding value. Formally: $
v^*=\max \{v \mid \exists (\tau(h),v) \in \CX,  h=H(j||v||req)\}
$,
$
req^*=\{req \mid \exists (\tau(h),v^*) \in \CX, h=H(j||v^*||req)  \}
$.
Otherwise, if $\CX$ contains no $(\tau(h),v)$ then set $v^* := -1$.

(2) We say that a value $req'$ is \emph{fast for $v$} if there exists $f+c+1$ messages in $\FX$ and for each such message $(\sigma_i(h),v) \in \FX$ it is the case that $h=H(j||v'||req')$ and $v' \geq v$.
Let $\hat{v}$ be the highest view such that there exists a value  $req'$ that is \emph{fast for $v$}. If its unique, let $\hat{req}$ the corresponding fast value for $\hat{v}$. Formally:
$
fast(req',v)=1 ~~ \mbox{iff} ~~ \exists M\subset \FX, |M|=f+c+1, \forall (\sigma_i(h),v') \in M,  h=H(j||v'||req') \wedge v' \geq v
$,
$
\hat{v} = \max \{v \mid \exists req' \mid fast(req',v)=1 \}
$,
$
\hat{req} = \{req' \mid fast(req',\hat{v})=1 \}
$.
If no such $\hat{v}$ exists or if for $\hat{v}$ there is more than one potential value $\hat{req}$ then set $\hat{v} := -1$ .

(3) If $v^* \geq \hat{v}$ and $v^*>-1$ then set $y := \langle \text{``pre-prepare''}, j, v+1, req^*, H(j||v+1||req^*) \rangle$.

Otherwise if $\hat{v} > v^* $ then set $y := \langle \text{``pre-prepare''}, j, v+1, \hat{req},H(j||v+1)||\hat{req}) \rangle$.

Otherwise set $y := \langle \text{``pre-prepare''}, j, v+1, \text{``null''}, H(j||v+1||\text{``null''})\rangle$, where $\text{``null''}$ is the no-op operation.

\subsubsection{Efficient view change via pipelining}
Currently \sysname allows committing a sequence number $x$ before the pre-prepare information for sequence numbers $<x$ have arrived. This allows \sysname a high degree of parallelism. This also means that like PBFT, during a view change \sysname needs to suggest a value for each sequence number between $\ls$ and $\ls+\window$.

An alternative approach is to commit sequence number $x$ only after all the pre-prepare messages of all sequences $\leq x$ have arrived and $x$'s hash is a commitment to the whole history. Concretely,  $h_x$ is a hash not of $(r||s||v)$ but of $(r||s||v||h_{x-1})$. This means that when a primary commits sequence number $x$ it is implicitly committing all the sequence numbers $\le x$ in the same decision.

With these changes, we can have a more efficient view change. Instead of sending pre-prepare values (with proof) for each sequence number from $\ls$ to $\ls+\window$, it is sufficient for the new primary to gather from each replica just two pairs (1) the first pair is $(h_j,v)$ where $v$ is the highest view for which the replica has $\tau(h_j)$ and $j$ is the slot number (assume $v:=-1$ if no such view exists); (2) the second pair is $(h'_j,v')$ where  $v'$ is the highest view for which the replica has $f+c+1$ pre-prepare messages with $h_j$ where $j$ is the slot number.

As before, the primary gathers $2f+2c+1$ such messages and chooses the highest view from $(v,h)$ and $(v',h')$ (preferring $(v,h)$ if there is a tie).

The advantage of this view change is that just two values are sent irrespective of the size of the window.


\section{Safety}

Safety is captured in the following Theorem:

\begin{theorem}\label{thm:safe}
If any two non-faulty replicas commit on a decision block for a given sequence number then they both commit on the same decision block.
\end{theorem}

Fix a sequence number $j$ and let  $v'$ be smallest view number at which some non-faulty replica commits on a decision block. Let the hash of this decision block on sequence number $j$ be $h$ where $h=H(j||v'||req)$. 
There are two cases to consider: (1) at least one of the  non-faulty that committed at view $v'$ committed due to a signature $\tau(\tau(h))$ or (2) all the committing non-faulty at $v'$ committed due to a signature $\sigma(h)$. We will now prove Theorem \ref{thm:safe} by looking at each cases separately.

\begin{lemma} 
If a non-faulty commits to $req$ on sequence $j$ due to $\tau(\tau(h))$ at view $v'$ then there exists a set $SC$ of $f+c+1$ non-faulty replicas such that for any view $v \geq v'$:
\begin{enumerate}
    \item For each replica $r \in SC$, at view $v$, replica $r$'s highest commit proof $lm_j=\tau(\tau(h))$ is for hash $h=H(j||v''||req)$ such that $v \geq v'' \geq v'$.
    
    \item Any valid signature $\tau(h')$ where $h'=H(j||v''||req'')$ generated by the adversary has the property that either $req''=req$ or $v''<v'$.
    
    \item If $v=v'$ then there are at most $f+c$ non-faulty that accepted a pre-prepare with some $req' \neq req$ at view $v'$.
    
    \item If $v>v'$ then there are no non-faulty replicas that accepted a pre-prepare with some $req' \neq req$ at view $v'$.
    
\end{enumerate}
\end{lemma}

\begin{proof}
The proof is by induction on $v$ for all $v\geq v'$. We start with the base case of $v=v'$. Property 1 holds since some non-faulty committed after seeing $2f+c+1$ commit messages on $h=H(j||v'||req)$, of which at least $f+c+1$ came from non-faulty. Fix $SC$ to be this set of non-faulty. Property 2 and 3 hold since
in view $v'$ at least $f+c+1$ non-faulty send a pre-prepare for $h$ and a non-faulty will send at most one pre-prepare per view, so there can be at most $f+c$ non-faulty that sent a pre-prepare for $req' \neq req$ at view $v=v'$. Finally, note that property 4 is vacuously true.

Assume the properties hold by induction on all views smaller than $v$ and consider view $v$. From property 1 of the induction hypothesis on view $v-1$, any view change to view $v$ must include at least one message $\tau(h)$ with $h=H(j||v''||req)$ from $SC$. This is true because any view change set is of size $2f+2c+1$ and so it must intersect the set $SC$ from property 1. From property 2, $v''$ is larger than any other $\tau(h')$ for $req' \neq req$ included in any view change set. From properties 3 and 4 it must be the case that if there are $f+c+1$ replicas that send pre-prepare for for $req' \neq req$ then at least $c+1$ of them are honest and hence must be of view at most $v'$.

Recall that in the view change protocol we define that $req'$ is fast for view $u$ if there exists $f+c+1$ messages in $\FX$ and for each such message $(\sigma_i(z),u) \in \FX$ it is the case that $z=H(j||v'||req')$ and $v' \geq u$. Therefore it must be the case that if $req'$ is fast for view $u$ then $u \leq v'$. Also recall that in the view change protocol $v*$ is defined to be the highest view with a prepare message $\tau(h^*)$ in $LX$, so we have $v* \geq v'$. Note that even if $v* = v'$ then the outcome of the view change protocl (the value $y_j$) will use $req$ over $req'$ (this is where our view change algorithm prefers the slow path proof over the fast path proof).

Together this means that the only possible outcome of the view change will be to set $y_j$ for view $v$ to have the value $req$. This is true because property 1 shows that $req$ will be seen in a prepare message, property 2 that it will be associated with the maximal prepare, and property 3+4 that if $\hat{req}$ is chosen then it must be that $\hat{req}=req$. This proves property 4 (property 3 is vacuously true).

Given property 4 on view $v$ we now prove properties 1 and 2 for view $v$. Since the only safe value for view $v$ is $req$ then the only change in $SC$ is that a replica may update its highest view to $v$ (if the Primary in $v$ manages to send a prepare message in view $v$) but this must be for the same block content $req$ so property 1 holds. Similarly property 2 holds because no non-faulty we accept in $v$ a pre-prepare for $req' \neq req$ and hence no such $\tau(h')$ can be generated.
\end{proof}

We now do a similar analysis for the case that all committing non-faulty at $v'$  do it due to a signature $\sigma(h)$. Importantly, property 3 below is slightly stronger than property 3 above.

\begin{lemma} 
If a non-faulty commits to $req$ on sequence $j$ due to $\sigma(h)$ in view $v'$ then there exists a set $FC$ of $2f+c+1$ non-faulty replicas such that for any view $v \geq v'$:
\begin{enumerate}
    \item For each replica $r \in FC$, in view $v$, the highest view $v''$ for which a pre-prepare was accepted with hash $h=H(j||v''||req'')$ has the property that  $v'' \geq v'$ and $req''=req$.
    
    \item Any valid signature $\tau(h')$ where $h'=H(j||v''||req'')$ generated by the adversary has the property that either $req''=req$ or $v''<v'$.
    
    \item If $v=v'$ then there are at most $c$ non-faulty that accepted a pre-prepare with some $req' \neq req$ at view $v'$.
    
    \item If $v>v'$ then there are no non-faulty that accepted a pre-prepare with some $req' \neq req$ at view $v'$.

\end{enumerate}
\end{lemma}

\begin{proof}
The proof is by induction on $v$ for all $v\geq v'$. We start with the base case of $v=v'$. Property 1 holds since some non-faulty committed after seeing $3f+c+1$ pre-prepare messages on $h=H(j||v'||req)$, of which at least $2f+c+1$ came from non-faulty. Property 2 and 3 hold since
in view $v=v'$ at least $2f+c+1$ non-faulty send a pre-prepare for $h$ and a non-faulty will send at most one pre-prepare per view, so there can be at most $c$ non-faulty that sent a pre-prepare for some $req' \neq req$ (property 4 is vacuously true).

Assume the properties hold by induction on all views smaller than $v$ and consider view $v$. From the induction hypothesis property 1 on view $v-1$ we have that any view change to $v$ must include at least $f+c+1$ messages of the form $\sigma_i(h)$ with $h=H(j||v''||req)$ from $ i \in FC$. This is true because any view change set is of size $2f+2c+1$ and so it must intersect $FC$ with at least $f+c+1$ members. From the induction hypothesis on property 3 and 4 it follows that if there are $f+c+1$ replicas that send pre-prepare for $req' \neq req$  then it must be the case that at least \emph{one} of the $c+1$ honest must be of view that is strictly smaller than $v'$. This implies that if some $req' \neq req$ is a fast for $u$ then $u<v'$. Hence $req$ will be the unique value that is fast at view at least $v'$ and hence the view change protocol will set $\hat{req}=req$.

We still need to show that $\hat{req}=req$ will be chosen over any $req^* \neq req$. From the induction hypothesis on property $2$ it follows that for $req^* \neq req$ it must be that $v^* < v'$ and hence in the view change $\hat{v}$ will cause $\hat{req}=req$ to be selected.

Together this means that the only outcome of the view change will be to set $y_j$ for view $v$ to have the value $req$. This is true because property 1 shows that $req$ will be seen at fast for at least $v'$, property 2 that $v^*$  will be strictly smaller if $req^* \neq req$ , and property 3+4 imply that if $\hat{req}$ is chosen then it must be that $\hat{req}=req$. This proves property 4 (property 3 is vacuously true).

Given property 4 on $v$ we now prove properties 1 and 2. Since the only safe value for $v$ is $req$ then the only change in $SC$ is that replicas may update their highest view of a pre-prepare to view $v$ but this must be for $req$ so property 1 holds. Similarly property 2 holds because no non-faulty we accept in $v$ a pre-prepare for $req' \neq req$ and hence no such $\tau(h')$ can be generated.
\end{proof}

Theorem \ref{thm:safe}  follows directly from the  the two lemma above.

\section{Liveness}

\sysname is deterministic so lacks liveness in the asynchronous mode, due to FLP~\cite{FLP85}.
As in PBFT~\cite{CL99}, liveness is obtained by striking a balance between making progress in the current view and moving to a new view.
\sysname uses the techniques of PBFT~\cite{CL99} tailored to a larger deployment: (1) exponential back-off view change timer; (2) replica issues a view change if it hears $f+1$ replicas issue a view change; (3) a view can continue making progress even if $f$ or less replicas send a view change. Finally \sysname uses $c+1$ collectors to make progress in the fast path and in the common path we ensure that one of the collectors is the primary.

 Unlike some protocols, our protocol can always wait for at most $n-f$ messages to make progress both in the common path and in the view change (for example in PBFT journal version you may need to wait for more messages in the view change protocol). Hence not only is our protocol clearly deadlock free, it is also \emph{reactive}, meaning that after GST it makes progress at the speed of the fastest $n-f$ replicas and does not need to wait for the maximum network delay.

We still need to show that progress is made after GST with a non-faulty primary. Again this is quite strait forward and follows from the fact that in the common mode the Primary is also a collector.

Finally we note that the liveness of the view change protocol follows from the following pattern: the primary makes a decision based on signed messages (its proof) and then forwards both the decision and the signed messages (its proof) so all replicas can repeat exactly the same computation.

\section{\sysname Implementation}\label{sec:implementation}

\sysname is implemented in C++ and follows some parts of the design of the original PBFT code~\cite{PBFT-lib,CL99,CL02}, in particular PBFT's \textit{state transfer} mechanism. \sysname has been in active development for over two years.

\textbf{Cryptography implementation}
Cryptographic primitives (RSA 2048, SHA256, HMAC) are implemented using the Crypto++ library~~\cite{CryptoCPP2016}. 
To implement threshold BLS, we use RELIC \cite{relic-toolkit}, a cryptographic library with support for pairings.
We use the BN-P254 \cite{bn-p254} elliptic curve, which provides the same security as 2048-bit RSA (i.e., 110-bit security) even with recent developments on discrete-log attacks \cite{barbulescu-key-size,menezes-pairing-security}.
To reduce latency associated with combining threshold BLS based shares (in the collectors) we  parallelized the independent exponentiations and use a background thread.
In the fast path, as long as no failure is detected, we use a BLS group signature ($n$-out-of-$n$ threshold) which provides smaller latency than BLS threshold signatures. We implemented a mechanism to automatically switch to and from group signatures and threshold signatures based on recent history. 

\textbf{\sysname batching and parallelism parameters}
We use an adaptive leaning algorithm that dynamically modifies the size of the parameter  $\batch$ based on the number of currently running concurrent sequence numbers. The $\batch$ size is the minimum number of client operations in each block.
The number of decision blocks that can be committed in parallel is $\window = 256$. The value $active\text{-}window= \lfloor (n-1)/(c+1) \rfloor$ is the actual number of decision blocks that are committed in parallel by the primary.

\textbf{Blockchain smart contract implementation}
The EVM implementation we used is based on cpp-ethereum~\cite{CPPEth}. 
We integrated storage-related commands with our key-value store interface and use RocksDB~\cite{RocksDB} as its backend.

\section{Performance Evaluation}

\begin{figure*}[!ht]
\begin{minipage}{0.005\paperwidth}
{\scriptsize \rotatebox{90}{Throughput (operations/second)}}
\end{minipage}
\begin{minipage}{0.990\paperwidth}
	\begin{center}
\begin{minipage}{1.0\textwidth}		
		\begin{tikzpicture}
		\begin{axis}[throughputGraphStyle,plotSetA,title={\scriptsize \textbf{no failures}},width=0.29\linewidth,ylabel = {\scriptsize \textbf{batch{=}64}}]
		\addplot [blues5,mark=square*,style={dashed}
] table [x={clients}, y={PBFTThr}]
		{Figures/ThroughputF64Fails0Batch64.txt}; 
		\addplot [antiquefuchsia,mark=oplus*,style={dashed}] table [x={clients}, y={LinearPBFTThr}]
{Figures/ThroughputF64Fails0Batch64.txt};
		\addplot [black,mark=oplus] table [x={clients}, y={SBFTNoMerThr}]
		{Figures/ThroughputF64Fails0Batch64.txt}; 
		\addplot [red,mark=triangle*] table [x={clients}, y={SBFTCZeroThr}]
		{Figures/ThroughputF64Fails0Batch64.txt};
		\addplot [darkspringgreen,mark=diamond*] table [x={clients}, y={SBFTCEightThr}]
		{Figures/ThroughputF64Fails0Batch64.txt};
\addlegendentry{\cPBFT};
\addlegendentry{\cLinearPBFT};
\addlegendentry{\cSBFTTagA};
\addlegendentry{\cSBFTCZero};
\addlegendentry{\cSBFTCEight};
		\end{axis}
		\end{tikzpicture}
		\begin{tikzpicture}
		\begin{axis}[throughputGraphStyle,title={\scriptsize \textbf{8 failures}},width=0.29\linewidth]
		\addplot [blues5,mark=square*,style={dashed}
		] table [x={clients}, y={PBFTThr}]
		{Figures/ThroughputF64Fails8Batch64.txt}; 
		\addplot [antiquefuchsia,mark=oplus*,style={dashed}] table [x={clients}, y={LinearPBFTThr}]
		{Figures/ThroughputF64Fails8Batch64.txt};
		\addplot [black,mark=oplus] table [x={clients}, y={SBFTNoMerThr}]
		{Figures/ThroughputF64Fails8Batch64.txt}; 
		\addplot [red,mark=triangle*] table [x={clients}, y={SBFTCZeroThr}]
		{Figures/ThroughputF64Fails8Batch64.txt};
		\addplot [darkspringgreen,mark=diamond*] table [x={clients}, y={SBFTCEightThr}]
		{Figures/ThroughputF64Fails8Batch64.txt};
		\end{axis}
		\end{tikzpicture}
		\begin{tikzpicture}
		\begin{axis}[throughputGraphStyle,title={\scriptsize \textbf{64 failures}},width=0.29\linewidth]
		\addplot [blues5,mark=square*,style={dashed}
		] table [x={clients}, y={PBFTThr}]
		{Figures/ThroughputF64Fails64Batch64.txt}; 
		\addplot [antiquefuchsia,mark=oplus*,style={dashed}] table [x={clients}, y={LinearPBFTThr}]
		{Figures/ThroughputF64Fails64Batch64.txt};
		\addplot [black,mark=oplus] table [x={clients}, y={SBFTNoMerThr}]
		{Figures/ThroughputF64Fails64Batch64.txt}; 
		\addplot [red,mark=triangle*] table [x={clients}, y={SBFTCZeroThr}]
		{Figures/ThroughputF64Fails64Batch64.txt};
		\addplot [darkspringgreen,mark=diamond*] table [x={clients}, y={SBFTCEightThr}]
		{Figures/ThroughputF64Fails64Batch64.txt};
		\end{axis}
		\end{tikzpicture}
\end{minipage}		
\begin{minipage}{1.0\textwidth}		
	\begin{tikzpicture}
	\begin{axis}[throughputGraphStyle,width=0.29\linewidth,ylabel = {\scriptsize \textbf{no batch}}]
	\addplot [blues5,mark=square*,style={dashed}
	] table [x={clients}, y={PBFTThr}]
	{Figures/ThroughputF64Fails0Batch1.txt}; 
	\addplot [antiquefuchsia,mark=oplus*,style={dashed}] table [x={clients}, y={LinearPBFTThr}]
	{Figures/ThroughputF64Fails0Batch1.txt};
	\addplot [black,mark=oplus] table [x={clients}, y={SBFTNoMerThr}]
	{Figures/ThroughputF64Fails0Batch1.txt}; 
	\addplot [red,mark=triangle*] table [x={clients}, y={SBFTCZeroThr}]
	{Figures/ThroughputF64Fails0Batch1.txt};
	\addplot [darkspringgreen,mark=diamond*] table [x={clients}, y={SBFTCEightThr}]
	{Figures/ThroughputF64Fails0Batch1.txt};
	\end{axis}
	\end{tikzpicture}
	\begin{tikzpicture}
	\begin{axis}[throughputGraphStyle,width=0.29\linewidth]
	\addplot [blues5,mark=square*,style={dashed}
	] table [x={clients}, y={PBFTThr}]
	{Figures/ThroughputF64Fails8Batch1.txt}; 
	\addplot [antiquefuchsia,mark=oplus*,style={dashed}] table [x={clients}, y={LinearPBFTThr}]
	{Figures/ThroughputF64Fails8Batch1.txt};
	\addplot [black,mark=oplus] table [x={clients}, y={SBFTNoMerThr}]
	{Figures/ThroughputF64Fails8Batch1.txt}; 
	\addplot [red,mark=triangle*] table [x={clients}, y={SBFTCZeroThr}]
	{Figures/ThroughputF64Fails8Batch1.txt};
	\addplot [darkspringgreen,mark=diamond*] table [x={clients}, y={SBFTCEightThr}]
	{Figures/ThroughputF64Fails8Batch1.txt};
	\end{axis}
	\end{tikzpicture}
	\begin{tikzpicture}
	\begin{axis}[throughputGraphStyle,width=0.29\linewidth]
	\addplot [blues5,mark=square*,style={dashed}
	] table [x={clients}, y={PBFTThr}]
	{Figures/ThroughputF64Fails64Batch1.txt}; 
	\addplot [antiquefuchsia,mark=oplus*,style={dashed}] table [x={clients}, y={LinearPBFTThr}]
	{Figures/ThroughputF64Fails64Batch1.txt};
	\addplot [black,mark=oplus] table [x={clients}, y={SBFTNoMerThr}]
	{Figures/ThroughputF64Fails64Batch1.txt}; 
	\addplot [red,mark=triangle*] table [x={clients}, y={SBFTCZeroThr}]
	{Figures/ThroughputF64Fails64Batch1.txt};
	\addplot [darkspringgreen,mark=diamond*] table [x={clients}, y={SBFTCEightThr}]
	{Figures/ThroughputF64Fails64Batch1.txt};
	\end{axis}
	\end{tikzpicture}
\end{minipage}				
	\end{center}
\end{minipage}
	\vspace{-3mm} 
	\begin{center}
		{ \scriptsize Number of clients }
	\end{center}
	\begin{center}
		\ref{plotSetALegened}
	\end{center}
    \vspace{-7mm}
	\caption{Throughput per clients} 
	\label{fig:throughput}
\end{figure*}
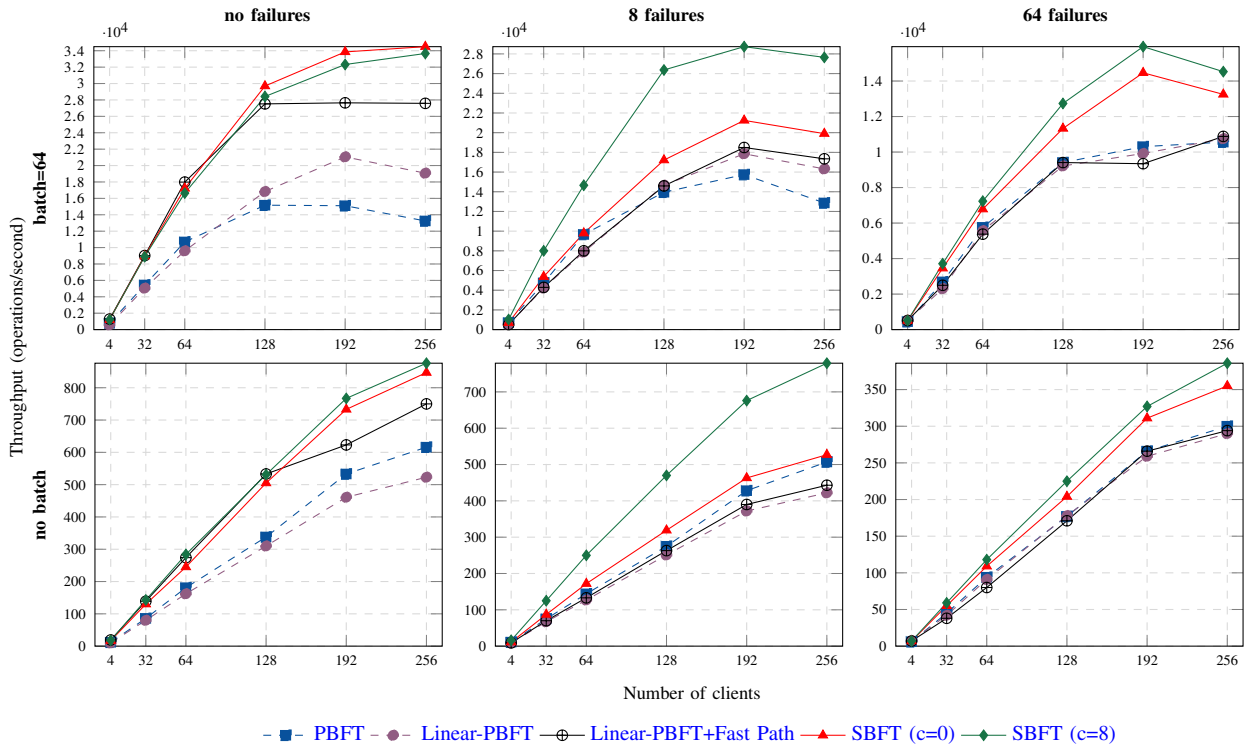

\begin{figure*}[!ht]
	\begin{minipage}{0.005\paperwidth}
		{\scriptsize \rotatebox{90}{Latency (milliseconds)}}
	\end{minipage}
	\begin{minipage}{0.990\paperwidth}
		\begin{center}
			\begin{minipage}{1.0\textwidth}		
		\begin{tikzpicture}
\begin{axis}[LatencyGraphStyle,scaled x ticks=false, plotSetB,title={\scriptsize \textbf{no failures}}, width=0.29\linewidth,ylabel = {\scriptsize \textbf{batch{=}64}}]
\addplot [blues5,mark=square*,style={dashed}
] table [x={PBFTThr}, y={PBFTLat}]
{Figures/ThroughputF64Fails0Batch64.txt}; 
\addplot [antiquefuchsia,mark=oplus*,style={dashed}] table [x={LinearPBFTThr}, y={LinearPBFTLat}]
{Figures/ThroughputF64Fails0Batch64.txt};
\addplot [black,mark=oplus] table [x={SBFTNoMerThr}, y={SBFTNoMerLat}]
{Figures/ThroughputF64Fails0Batch64.txt}; 
\addplot [red,mark=triangle*] table [x={SBFTCZeroThr}, y={SBFTCZeroLat}]
{Figures/ThroughputF64Fails0Batch64.txt};
\addplot [darkspringgreen,mark=diamond*] table [x={SBFTCEightThr}, y={SBFTCEightLat}]
{Figures/ThroughputF64Fails0Batch64.txt};		
\end{axis}
\end{tikzpicture}		
\begin{tikzpicture}
\begin{axis}[LatencyGraphStyle,scaled x ticks=false,title={\scriptsize \textbf{8 failures}},width=0.29\linewidth]
\addplot [blues5,mark=square*,style={dashed}
] table [x={PBFTThr}, y={PBFTLat}]
{Figures/ThroughputF64Fails8Batch64.txt}; 
\addplot [antiquefuchsia,mark=oplus*,style={dashed}] table [x={LinearPBFTThr}, y={LinearPBFTLat}]
{Figures/ThroughputF64Fails8Batch64.txt};
\addplot [black,mark=oplus] table [x={SBFTNoMerThr}, y={SBFTNoMerLat}]
{Figures/ThroughputF64Fails8Batch64.txt}; 
\addplot [red,mark=triangle*] table [x={SBFTCZeroThr}, y={SBFTCZeroLat}]
{Figures/ThroughputF64Fails8Batch64.txt};
\addplot [darkspringgreen,mark=diamond*] table [x={SBFTCEightThr}, y={SBFTCEightLat}]
{Figures/ThroughputF64Fails8Batch64.txt};		

\end{axis}
\end{tikzpicture}		
\begin{tikzpicture}
\begin{axis}[LatencyGraphStyle,scaled x ticks=false,title={\scriptsize \textbf{64 failures}},width=0.29\linewidth]
\addplot [blues5,mark=square*,style={dashed}
] table [x={PBFTThr}, y={PBFTLat}]
{Figures/ThroughputF64Fails64Batch64.txt}; 
\addplot [antiquefuchsia,mark=oplus*,style={dashed}] table [x={LinearPBFTThr}, y={LinearPBFTLat}]
{Figures/ThroughputF64Fails64Batch64.txt};
\addplot [black,mark=oplus] table [x={SBFTNoMerThr}, y={SBFTNoMerLat}]
{Figures/ThroughputF64Fails64Batch64.txt}; 
\addplot [red,mark=triangle*] table [x={SBFTCZeroThr}, y={SBFTCZeroLat}]
{Figures/ThroughputF64Fails64Batch64.txt};
\addplot [darkspringgreen,mark=diamond*] table [x={SBFTCEightThr}, y={SBFTCEightLat}]
{Figures/ThroughputF64Fails64Batch64.txt};		
\end{axis}
\end{tikzpicture}	
			\end{minipage}		
			\begin{minipage}{1.0\textwidth}		
		\begin{tikzpicture}
\begin{axis}[LatencyGraphStyle,scaled x ticks=false,plotSetB, width=0.29\linewidth,ylabel = {\scriptsize \textbf{no batch}}]
\addplot [blues5,mark=square*,style={dashed}
] table [x={PBFTThr}, y={PBFTLat}]
{Figures/ThroughputF64Fails0Batch1.txt}; 
\addplot [antiquefuchsia,mark=oplus*,style={dashed}] table [x={LinearPBFTThr}, y={LinearPBFTLat}]
{Figures/ThroughputF64Fails0Batch1.txt};
\addplot [black,mark=oplus] table [x={SBFTNoMerThr}, y={SBFTNoMerLat}]
{Figures/ThroughputF64Fails0Batch1.txt}; 
\addplot [red,mark=triangle*] table [x={SBFTCZeroThr}, y={SBFTCZeroLat}]
{Figures/ThroughputF64Fails0Batch1.txt};
\addplot [darkspringgreen,mark=diamond*] table [x={SBFTCEightThr}, y={SBFTCEightLat}]
{Figures/ThroughputF64Fails0Batch1.txt};		
\end{axis}
\end{tikzpicture}		
\begin{tikzpicture}
\begin{axis}[LatencyGraphStyle,scaled x ticks=false,width=0.29\linewidth]
\addplot [blues5,mark=square*,style={dashed}
] table [x={PBFTThr}, y={PBFTLat}]
{Figures/ThroughputF64Fails8Batch1.txt}; 
\addplot [antiquefuchsia,mark=oplus*,style={dashed}] table [x={LinearPBFTThr}, y={LinearPBFTLat}]
{Figures/ThroughputF64Fails8Batch1.txt};
\addplot [black,mark=oplus] table [x={SBFTNoMerThr}, y={SBFTNoMerLat}]
{Figures/ThroughputF64Fails8Batch1.txt}; 
\addplot [red,mark=triangle*] table [x={SBFTCZeroThr}, y={SBFTCZeroLat}]
{Figures/ThroughputF64Fails8Batch1.txt};
\addplot [darkspringgreen,mark=diamond*] table [x={SBFTCEightThr}, y={SBFTCEightLat}]
{Figures/ThroughputF64Fails8Batch1.txt};		

\end{axis}
\end{tikzpicture}		
\hspace{3mm}
\begin{tikzpicture}
\begin{axis}[LatencyGraphStyle,scaled x ticks=false,width=0.29\linewidth]
\addplot [blues5,mark=square*,style={dashed}
] table [x={PBFTThr}, y={PBFTLat}]
{Figures/ThroughputF64Fails64Batch1.txt}; 
\addplot [antiquefuchsia,mark=oplus*,style={dashed}] table [x={LinearPBFTThr}, y={LinearPBFTLat}]
{Figures/ThroughputF64Fails64Batch1.txt};
\addplot [black,mark=oplus] table [x={SBFTNoMerThr}, y={SBFTNoMerLat}]
{Figures/ThroughputF64Fails64Batch1.txt}; 
\addplot [red,mark=triangle*] table [x={SBFTCZeroThr}, y={SBFTCZeroLat}]
{Figures/ThroughputF64Fails64Batch1.txt};
\addplot [darkspringgreen,mark=diamond*] table [x={SBFTCEightThr}, y={SBFTCEightLat}]
{Figures/ThroughputF64Fails64Batch1.txt};		
\end{axis}
\end{tikzpicture}		
			\end{minipage}				
		\end{center}
	\end{minipage}
	\vspace{-3mm} 
	\begin{center}
		{ \scriptsize Throughput (operations/second) }
	\end{center}
	\begin{center}
		\ref{plotSetALegened}
	\end{center}
    \vspace{-7mm}
    \caption{Latency vs Throughput.} 
    
    \label{fig:LatencyVsThroughput}	
\end{figure*}
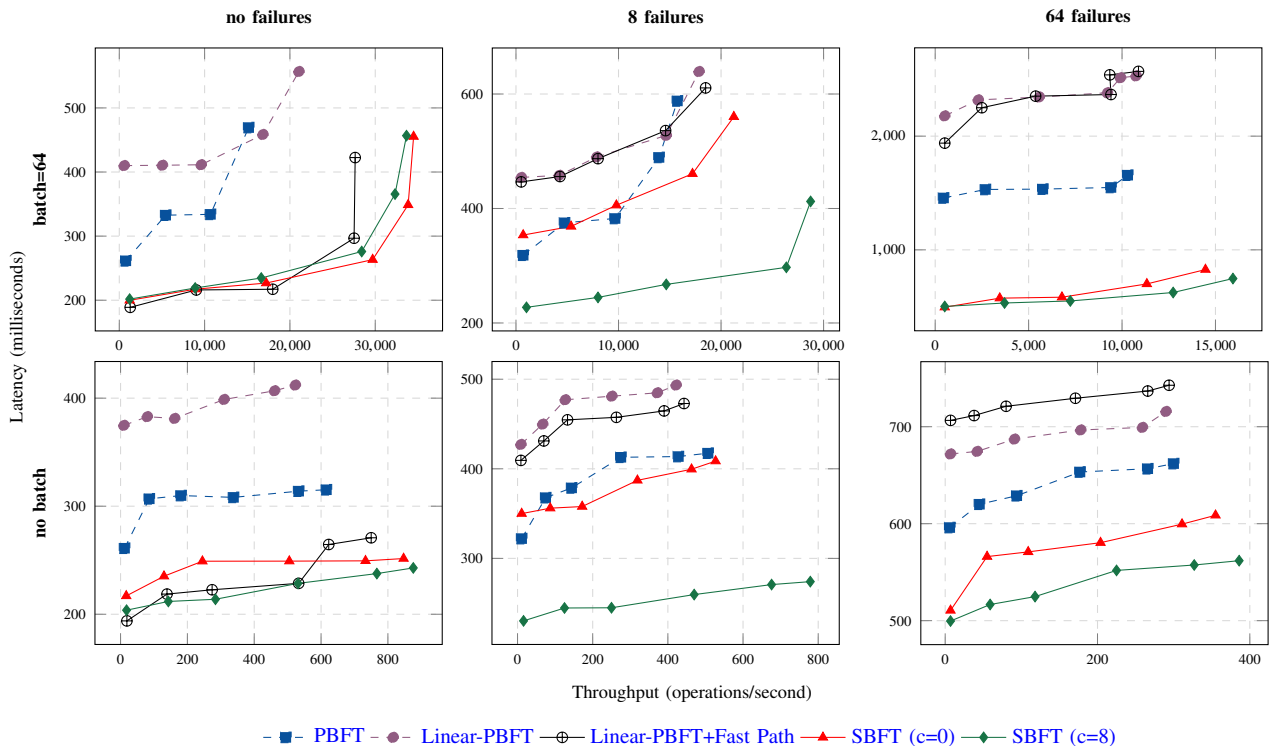

We evaluate \sysname by deploying $~200$  replicas over a wide area geo-distributed network. All experiments are configured to withstand $f=64$ Byzantine failures. %
Following, \cite{AAA09} \sysname uses public-key signed client requests and server messages.

At the time of the experiments, we could not find other BFT implementations that (1) was freely available online; and  (2)
could reliably work on a real (not simulated) world scale WAN and withstand $f=64$ failures.
The freely available code for PBFT could not scale and was not updated in the last 10 years. The code for Zyzzyva \cite{KAD09} contains a serious safety violation \cite{revisiting17} and did not contain a state transfer module.
Both BizCoin \cite{KJ16,cosi} and Omniledger \cite{omniledger}  have GitHub repositories but have only reported simulation results. Other projects like Algorand \cite{algo17} have only simulations and no open source code. Moreover the focus on these systems is on permisionless models using proof-of-work or proof-of-stake not the standard permissioned model. Comparing to these systems is left as future work (once there is a freely accessible version that is robust enough to be readily deployed and support EVM smart contracts).

Our goal is to report and compare in a real world deployment on a Wide Area Network that actually persists transactions to disk and executes real world EVM smart contracts.  
We therefore spent several months significantly improving, fixing and hardening an existing PBFT code-base in order to make it reliably work in our experimental setting. We call this implementation \textit{scale optimized PBFT}. 

We note that Souse et al \cite{HL-BFT18} use an implementation of PBFT called BFT-SMaRt \cite{BFT-SMART14}. However, it seems that the WAN deployment reported in \cite{HL-BFT18} scales to only $f \leq 3$ in a LAN and $f=1$ for WAN. Our baseline scale optimized PBFT is tuned to provide better scalability. Moreover it was not clear to us how to run EVM contracts on top of this system.

In our experiments we start with a scale optimized PBFT implementation and then show how each of the 4 ingredients improves performance as follows:
(1) linear PBFT reduces communication and improves throughput at the cost of latency; (2) adding a fast path reduces latency; (3) using cryptography to allow a single message acknowledgement improves performance when there are many clients; (4) adding redundant servers to improve resilience improves the latency-throughput trade-off.

For micro-benchmarks we run a simple Key-Value service. For our main evaluation we use real transactions from Ethereum that are executed and committed to disk (via RocksDB). We take half a Million Ethereum transactions, spanning a time of 2 months, which included $\sim 5000$ contracts created.

We compare the following replication protocols:

\noindent (1) \textbf{PBFT} (the baseline): A scale optimized implementation of PBFT.

\noindent  (2) \textbf{Linear-PBFT} (adding ingredient 1): A modification of PBFT that avoids quadratic communication by using a collector.

\noindent (3) \textbf{Fast Path + Linear PBFT} (adding ingredients 1 and 2): Linear-PBFT with an added fast-path.

\noindent (4) \textbf{\sysname with $c{=}0$} (adding ingredients 1,2, and 3): Linear-PBFT with an added fast-path and an execution collector that allows clients to receive signed message acknowledgements.

\noindent (5) \textbf{\sysname with $c{=}8$} (adding all 4 ingredients): Adding redundant servers to better adapt to network variance and failures.

\textbf{Continent scale WAN.}
In this scenario we spread the replicas and clients across 5 different regions in the same continent. In each region we use two availability zones and in each zone we deploy one machine with 32 VCPUs, Intel Broadwell E5-2686v4 processors with clock speed 2.3 GHz and connected via a 10 Gigabit network. 

We deployed more than one replica or client into a single machine. This was done due to economic and logistic constraints.
One may wonder if the fact that we packed multiple replicas into a single machine significantly modified our performance measurements. To assess this we repeated our experiments once with 10 machines (1 per availability zone, each machine had about 20 replica VMs) and then with 20 machines (2 per availability zone, each machine had about 10 replica VMs). The results of these experiments were almost the same. We conclude that the effects of communication delays between having 10 or 20 machines have marginal impact in a world scale WAN. Not surprisingly, our experiments show that in a world scale WAN, performance depends at least on the median latency and that having $10-20\%$ of replicas with a much lower latency does not modify or increase performance.  

\textbf{World scale WAN.}
In this scenario we spread the replicas and clients across 15 regions spread over all continents. In each region we deploy one machine (we also tested running two machines per region with similar results).

\textbf{Measurements}
\textit{Key-Value store benchmark}: each client sequentially sends 1000 requests. In the \textit{no batching} mode each request is a \textit{single} \textit{put} operation for writing a random value to a random key in the Key-Value store. In the \textit{batching} mode each request contains 64 operations. This models a reasonable smart contract workload. Replicas execute operations by changing their state and does a re-sync using the state transfer protocol (see Section~\ref{sec:implementation} and ~\cite{CL02}) if it falls behind. We ran these experiments on a continent scale WAN.

\textit{Smart-Contract benchmark}: we used 500,000 real transactions from Ethereum to test the \sysname ledger protocol. Replicas execute each contract by running the EVM byte-code and persisting the state on-disk. Each client sends operations by batching transactions into chunks of 12KB (on average about 50 transactions per batch). We ran these experiments on both a continent scale WAN and a world scale WAN.

\textbf{Key-Value benchmark evaluation.}
The results of the Key-Value benchmark are shown in Figures~\ref{fig:throughput} and~\ref{fig:LatencyVsThroughput}.
We first observe that compared to our scale optimized PBFT, the linear-PBFT protocol provides better throughput (2k per sec vs 1.5k per sec) when the system is under load (128 to 256 clients) with batching. Smaller effects appear also in the no batching case. We conclude that reducing the communication from quadratic to linear by using a collector significantly improves throughput at some cost in latency.

To improve latency we then add a fast path to linear-PBFT. This significantly increases throughput to 2.8k per sec. As expected, in the no failure executions the fast path seems to improve both latency and throughput, but does not help when there are failures.    

We then see that adding an execution collector that allows clients to receive just one message (instead of $f+1$) for acknowledgment significantly improves performance (latency throughput trade-off) in all scenarios (no failures and with failures). This shows that the communication from servers to clients is a significant performance bottleneck.

Finally, by parameterizing \sysname for $c=8$ we show the effect of adding redundant servers. Not surprisingly, this makes a big impact when there are $f=8$ failures. In addition, we see significant advantages in the $f=0$ and $f=64$ cases. This is probably because adding redundancy  reduces the variance and effects of slightly slow servers or staggering network links.

\textbf{Smart-Contract benchmark evaluation.}
In the continent scale WAN experiment \sysname measured 378 transaction per second with a median latency of 254 milliseconds. In the same setting our scale optimized PBFT obtained just 204 transaction per second with a median latency of 538 milliseconds. We conclude that in the content-scale \sysname simultaneously provides 2x better latency and almost 2x better throughput.

In the world scale WAN experiments \sysname obtained 172 transaction per second with a median latency of 622 milliseconds. Scale optimized PBFT obtained 98 transaction per second with a median latency of 934 milliseconds. We conclude that in a world-scale \sysname simultaneously provides almost 2x better throughput and about 1.5x better latency.
 
We note that just executing these smart contracts on a single computer (and committing the results to disk) without running any replication provides a 840 transaction per second base line. We conclude that adding a continent scale WAN 200 node replication, \sysname obtains a 2x slowdown relative to the base line. Adding a world-scale WAN 200 node replication, \sysname obtains a 5x slowdown relative to the base line.

\section{Related work}\label{sec:related}

Byzantine fault tolerance was first suggested by Lamport et al.~\cite{L82}.  Rampart~\cite{R95} was one of pioneering systems to consider Byzantine fault tolerance for state machine replication~\cite{L98}. 

\sysname is based on many advances aimed at making Byzantine fault tolerance practical. PBFT~\cite{CL99,CL02} and the extended framework of BASE~\cite{BASE01} provided many of the foundations, frameworks, optimizations and techniques on which \sysname is built. \sysname uses the conceptual approach of separating commitment from execution that is based on Yin et al.~\cite{YMVA03}. Our linear message complexity fast path is based on the techniques of Zyzzyva~\cite{K07,KAD09} and its theoretical foundations~\cite{MA06}. Our use of a hybrid model that provides better properties for $c \leq f$ failures is inspired by the parameterized model of Martin and Alvisi~\cite{MA06}. 
Up-Right~\cite{C09} studied a different model that assumes many omission failures and just a few Byzantine failures. Visigoth~\cite{P15} further advocates exploiting data center performance predictability and relative synchrony. XFT~\cite{XFT16} focuses on a model that limits the adversaries ability to control both asynchrony and malicious replicas.
In contrast, \sysname provides safety even in the fully asynchronous model when less than a third of replicas are malicious.
Prime~\cite{Prime11} adds additional pre-rounds so that clients can be guaranteed a high degree of fairness.
Our protocol provides the same type of weak fairness guarantees as in PBFT; we leave for further work the question of adding stronger fairness properties to \sysname.

A2M~\cite{A2M07}, TrInc~\cite{L09} and Veronese et al.~\cite{VCBL13} use secure hardware to obtain non-equivocation. They present a Byzantine fault tolerant replication that is safe in asynchronous models even when $n=2f+1$. CheapBFT \cite{CheapBFT12}  relies on an FPGA-based trusted subsystem to improve fault tolerance.
\sysname is a software solution and as such is bounded by  the $n\geq 3f+1$ lower bound~\cite{FLM86}.

Our use of public key cryptography (as opposed to MAC vectors) and threshold signatures follows the approach of \cite{CKS00} (also see~\cite{AAA09, ADKLDNOZ06, ADDKLNOZ10}).
We heavily rely on threshold BLS signatures~\cite{BLS04,threshold-bls}. Several recent systems mention they plan to use BLS threshold signatures \cite{cosi,randherd,coniks,omniledger}. To the best of our knowledge we are the first to deploy threshold BLS in a real system.

An alternative to the primary-backup based state machine replication approach is to use Byzantine quorum systems \cite{MR97} and make each client a proposer. This approach was taken by  QU~\cite{QU05} and HQ~\cite{HQ06} and provides very good scalability when write contention is low. \sysname follows the primary-backup paradigm that funnels multiple requests though a designated primary leader. This allows \sysname to benefit from batching which is crucial for throughput performance in large-scale multi-client scenarios.
 
Recent work is aimed at providing even better liveness guarantees. Honeybadger~\cite{HB16} is the first practical Byzantine fault tolerance replication system that leverages randomization to circumvent the FLP~\cite{FLP85} impossibility. Honeybadger and more recently BEAT \cite{BEAT18} provide liveness even when the network is fully asynchronous and controlled by an adversarial scheduler.
SPBT follows the DLS/Paxos/viewstamp-replication paradigm~\cite{DLS88,L98,LJ12} extended to Byzantine faults that guarantees liveness only when the network is synchronous.

Algorand \cite{algo17} provides a permissionless system that can support thousands of users and implements a BFT engine that chooses a random dynamic committee of roughly 2000 users. However, Algorand's scalability was only evaluated in a simulation of a wide area network. Even under best case no-failure simulation conditions, Algorand seems to provide almost 100x slower latency (60 seconds) relative to \sysname (600 milliseconds). \sysname is experimentally evaluated in a real world-scale geo-replicated wide area network deployment, executing real smart contracts, persisting their output to disk and testing scenarios with failures.

FastBFT \cite{FastBFT} shares many of the properties of \sysname. It also focuses on a linear version of PBFT and a single message client acknowledgement. FastBFT  decentralizes trust in a scalable way, but it relies on secure hardware and essentially centralizes its security assumptions by relying on the security of a single hardware vendor: Intel's SGX. FastBFT's performance is evaluated only on a local area network. \sysname assumes commodity hardware that does not rely on any single hardware vendor. \sysname is extensively evaluated in a real world-scale geo-distributed deployment.

\section{Conclusion}

We implemented \sysname, a state-of-the-art Byzantine fault tolerant replication library and experimentally validated that it provides significantly better performance for large deployments over a wide-area geo-distributed deployment. \sysname performs well when there are tens of malicious replicas, its performance advantage increases as the number of clients increases. 
 
We have learned that measuring real executions (not simulations) on hundreds of replicas over a world scale WAN and persisting real world smart contracts on disk is non-trivial and requires careful system tuning and engineering. Our experiments show that each one of our algorithmic ingredients improves the measured performance. For about two hundred replicas, the overall performance advantage show that \sysname simultaneously provides almost 2x better throughput and about 1.5x better latency relative to a highly optimized system that implements the PBFT protocol. 

We have shown that \sysname can be robustly deployed for hundreds of replicas and withstand tens of Byzantine failures. We believe that the advantage of linear protocols (like \sysname) over quadratic protocols will be even more profound at higher scales. Measuring real deployments of thousands of replicas that withstand hundreds of Byzantine failures is beyond the scope of this work.

\bibliographystyle{plain}

\end{document}